\documentclass[11pt]{amsart}

\usepackage{amsfonts,latexsym}
\usepackage{amsmath}
\usepackage{amstext}
\usepackage{amssymb}
\usepackage{graphicx}

\newcommand{\R}{\mathbb{R}}

\newcommand{\N}{\mathbb{N}}
\newcommand{\I}{\mathbb{I}}

\newtheorem{theorem}{Theorem}
\newtheorem{lemma}[theorem]{Lemma}
\newtheorem{proposition}[theorem]{Proposition}

\theoremstyle{definition}
\newtheorem{definition}{Definition}
\theoremstyle{remark}
\newtheorem{remark}{Remark}

\begin{document}

\title[Dynamic Triggering Mechanisms]{Dynamic Triggering Mechanisms for \\ Event-Triggered Control\\ (Full version)}


\author[Antoine Girard]{Antoine Girard}
\address{Laboratoire Jean Kuntzmann \\
Universit\'e de Grenoble \\
B.P. 53, 38041 Grenoble, France} \email{antoine.girard@imag.fr}

\thanks{This work was supported by the Agence Nationale de la Recherche (COMPACS project - ANR-13-BS03-0004).}


\maketitle


\begin{abstract}
In this paper, we present a new class of event triggering mechanisms for event-triggered control systems. This class is characterized by the introduction 
of an internal dynamic variable, which motivates the proposed name of dynamic event triggering mechanism. The stability of the resulting closed loop system is proved and the influence of design parameters on the decay rate of the Lyapunov function is discussed. For linear systems,
we  establish a lower bound on the inter-execution time as a function of the parameters. 
The influence of these parameters on a quadratic integral performance index is also studied.
Some simulation results are provided for illustration of the theoretical claims. 
\end{abstract}

\section{Introduction}

Cyber-physical systems often involve several control loops with shared computational and communication resources. 
Efficient usage of these resources is therefore a central issue in cyber-physical systems design. Traditional digital control techniques often assume that
controllers execute periodically, independently from the state of the system. This time-triggered paradigm may result in unnecessary high workloads when
computational and communication resources may be more usefully assigned to some other tasks.
These limitations have resulted in a recent regain of interest for event-triggered control (see e.g.~\cite{HJT2012} and the references therein).
In event-triggered control systems, the inputs of a system are not updated periodically but only when some events occur.   
Most of the recent effort has been devoted to the development
of systematic techniques for the design of event triggering mechanism (ETM) that can be used for the implementation of a given stabilizing feedback controller.
The most commonly used ETM typically consists of a static rule given on the state of the system as in~\cite{tabuada2007}.

In this paper, we present a new class of ETM that use an additional internal dynamic variable, which motivates the name of dynamic ETM. 
The use of internal dynamic variables in ETM can be found in several works such as~\cite{wang2008,mazo2010,seuret2011} where the proposed mechanisms
are equipped with internal clocks, or in~\cite{postoyan2011,postoyan2011a} where some of the proposed mechanisms uses an internal  dynamic variable which can also be seen as a clock (it is monotonic) whose growth rate depends on the state of the system.
In the present work, the internal variable is actually a filtered version of the signal used to trigger events in~\cite{tabuada2007}, it is generally non-monotone.

We consider the framework introduced in~\cite{tabuada2007} (same class of systems, same assumptions). 
The paper is organized as follows. In section~\ref{sec:2}, 
we introduce the class of dynamic  ETM in the general framework of nonlinear control systems.
We prove the asymptotic stability of the closed loop system. 
The influence of design parameters on the decay rate of the Lyapunov function is discussed and we show
that the guaranteed lower bound on inter-execution times (i.e. the time between two input updates)
using a dynamic ETM cannot be smaller than that obtained for the static ETM presented in~\cite{tabuada2007}. 
In section~\ref{sec:3}, we specialize our framework to the case of linear systems. 
We  establish a lower bound on the inter-execution time as a function of the parameters. 
The influence of these parameters on a quadratic integral performance index is also studied.
Finally, in section~\ref{sec:example}, we provide some simulation results
 for illustration of the theoretical developments. 

\paragraph*{Notation} $\R_0^+$ denotes the set of non-negative real numbers.
A function $\alpha: \R_0^+ \rightarrow \R_0^+$ is said to be of class $\mathcal K$ if it is continuous, strictly increasing and $\alpha(0)=0$;
if in addition $\alpha(r) \rightarrow +\infty$ as $r\rightarrow +\infty$, $\alpha$ is said to be of class $\mathcal K_\infty$. A function $\beta: \R_0^+\times \R_0^+ \rightarrow \R_0^+$ is said to be of class $\mathcal{KL}$ if it is continuous,
 $\beta(.,s) \in \mathcal K$ for all $s\in  \R^+$, and for all $r \in \R_0^+$,   $\beta(r,.)$ is strictly decreasing and $\beta(r,s)\rightarrow 0$ as $s\rightarrow +\infty$.
Given a function $g:\R_0^+ \rightarrow \R^n$, for all $t>0$, we denote by $g(t^-)$ the limit of $g(s)$ when $s$ approaches $t$ from the left.
Let us remark that if $g$ is continuous at $t$ then $g(t^-)=g(t)$.
 A function $f:\R^n\rightarrow \R^m$ is locally Lipschitz continuous, if for $x\in \R^n$, there exists a neighborhood of $x$, $S\subset \R^n$ and a constant $L>0$ such that
 $\|f(x_1)-f(x_2)\|\le L \|x_1-x_2\|$ for all $x_1,x_2\in S$; it is  
 Lipschitz continuous on compacts if for every compact set $S\subset \R^n$ there exists a constant $L>0$ such that $\|f(x_1)-f(x_2)\|\le L \|x_1-x_2\|$ for all $x_1,x_2\in S$.
\section{Triggering Mechanisms for Event-Triggered Control Systems}
\label{sec:2}

We consider the framework introduced in~\cite{tabuada2007} and deal with a control system of the form:
\begin{equation}
\label{eq:sys}
\dot x = f(x,u),\; x\in \R^n,\; u\in \R^m.
\end{equation}
It is assumed in the following that a feedback controller $k:\R^n \rightarrow \R^m$ has been designed such that the closed loop system
\begin{equation}
\label{eq:clsys}
\dot x = f(x,k(x+e))
\end{equation}
is input-to-state stable (ISS) with respect to measurement errors $e\in \R^n$.
As in~\cite{tabuada2007}, we assume that we know an ISS Lyapunov function for (\ref{eq:clsys}):
\begin{definition} A smooth function $V:\R^n\rightarrow \R^+_0$ is said to be an ISS-Lyapunov function for system (\ref{eq:clsys})
if there exist class $\mathcal K_\infty$ functions $\overline \alpha$, $\underline \alpha$, $\alpha$ and $\gamma$ satisfying
for all $x,e\in \R^n$ 
$$
\underline \alpha (\|x\|) \le V(x) \le \overline \alpha(\|x\|),
$$
$$
\nabla V(x) \cdot f(x,k(x+e)) \le -\alpha(\|x\|) + \gamma(\|e\|).
$$
\end{definition} 
We assume that the controller is implemented on a digital platform so that the actual control input of (\ref{eq:sys}) is given by
$$
u(t)=u(t_i)=k(x(t_i)), \; \forall t\in [t_i,t_{i+1}), \; i\in \I
$$
where the elements of the increasing sequence $(t_i)_{i\in \I}$ are the {\it execution times} at which the control input is computed and  updated. If there is an infinite number of executions, then $\I=\N$ and we denote $t_\infty$ as the limit of $t_i$ when $i\rightarrow +\infty$. 
If there is a finite number $I\in \N$ of executions, then $\I=\{0,1,\dots,I\}$ and we define $t_\infty=t_{I+1}=+\infty$.
Defining 
\begin{equation}
\label{eq:error}
e(t)=x(t_i)-x(t), \; \forall t\in [t_i,t_{i+1}), \; i\in \I
\end{equation}
the closed loop system is of the form (\ref{eq:clsys}).
In event-triggered control systems, the execution times are triggered by events that are generated according to a rule on the state of the system.
This rule is called the {\it event triggering mechanism} (ETM). 

\subsection{Static event triggering mechanisms}
In~\cite{tabuada2007}, it is proposed to generate an event whenever $\gamma(\|e(t^-)\|)=\sigma \alpha(\|x(t)\|)$ where
$\sigma\in (0,1)$ is a parameter.
Then,  the sequence $(t_i)_{i\in \I}$ is formally defined by
\begin{equation}
\label{eq:static}
\begin{array}{l}
t_0=0,\\
t_{i+1} = \inf\left\{\begin{array}{l} t\in \R |\; t>t_i \land \\
\sigma \alpha(\|x(t)\|)-\gamma(\|e(t^-)\|) \le 0  \end{array} \right\}
\end{array}
\end{equation}
It can be shown that using this ETM, the value of $\sigma \alpha(\|x\|) - \gamma(\|e\|)$ remains non-negative for all time. Then, it holds 
$$
\frac{d}{dt}V(x(t)) \le (\sigma-1)\alpha(\|x(t)\|),\; \forall t\in[0,t_\infty).
$$
which guarantees that $x(t)$ converges asymptotically to the origin provided $t_\infty=+\infty$.
In addition, if $\alpha \circ \overline\alpha^{-1}$ is locally Lipschitz continuous, then
it is straightforward to show by the Comparison Lemma (see e.g.~\cite[pages 102-103]{khalil}) that 
\begin{equation}
\label{eq:decaystatic}
V(x(t))  \le  \phi \left(V (x(0)),(1-\sigma)t \right),\; \forall t\in[0,t_\infty).
\end{equation}
where $\phi$ is the solution of the scalar differential equation
$\dot \phi(r,t) = - \alpha\circ \overline\alpha^{-1} ( \phi(r,t))$ with $\phi(r,0)=r$ (note that $\phi$ is a $\mathcal{KL}$ function by Lemma 4.4 in~\cite[page 145]{khalil})).
Hence, the decay rate of the Lyapunov function $V(x(t))$ can be adjusted using the parameter $\sigma$:
when $\sigma$ approaches $0$, the decay rate approaches that of the ``ideal'' system (when $e(t)=0$ for all $t\in \R_0^+$).

An important question in event triggered control is the existence of a {\it minimal inter-execution time}, which is some bound $\tau>0$ such that the sequence $(t_i)_{i\in \I}$ satisfies 
\begin{equation}
\label{eq:tau}
t_{i+1}-t_i \ge \tau, \; \forall i\in \I.
\end{equation}
Indeed, let us remark that if there exists $\tau$ such that (\ref{eq:tau}) holds then $t_\infty=+\infty$ and the stability of the closed loop system is proved. Moreover, this lower bound between two successive execution times provides us some requirements on digital platforms
on which the controller can be implemented.
In~\cite{tabuada2007}, under the assumption that  $f$, $k$, $\alpha^{-1}$ and $\gamma$ are Lipschitz continuous on compacts,
it is shown that for all $\sigma\in (0,1)$, for all compact set $S\subset \R^n$ containing the origin, there exists $\tau>0$ such that for all initial condition $x(0)\in S$, the sequence
$(t_i)_{i\in \I}$ defined by $(\ref{eq:static})$ satisfies (\ref{eq:tau}).

We call the ETM  $(\ref{eq:static})$, {\it static} because it only involves the current value of $x$ and $e$.
In the following, we propose {\it dynamic} ETM, that use an additional internal dynamic variable.

\subsection{Dynamic event triggering mechanisms}

We propose to enrich our ETM with an internal dynamic variable $\eta$ satisfying the following differential equation:
\begin{equation}
\label{eq:s}
\dot \eta = -\beta(\eta) + \sigma \alpha(\|x\|) - \gamma(\|e\|), \; \eta(0)=\eta_0.
\end{equation}
where the locally Lipschtiz continuous $\mathcal K_\infty$ function $\beta$ and the reals $\sigma\in (0,1)$ and $\eta_0 \in \R_0^+$ are design parameters.
Intuitively, $\eta$ can be regarded as a filtered value of $\sigma \alpha(\|x\|) - \gamma(\|e\|)$, where the filter (\ref{eq:s}) is possibly nonlinear if the function $\beta$ is nonlinear.
The main intuition behind the proposed dynamic ETM is that for stability of the closed loop system, it is not needed that  $\sigma \alpha(\|x\|) - \gamma(\|e\|)$ is always non-negative and it is sufficient that it is non-negative in average.
This can be ensured by triggering events in such a way that $\eta$ remains non-negative for all time. 
Then, let us consider the ETM defined by the following rule:
\begin{equation}
\label{eq:dynamic}
\begin{array}{l}
t_0=0,\\
t_{i+1} = \inf\left\{\begin{array}{l} t\in \R |\; t>t_i \land \\ 
\eta(t)+\theta(\sigma \alpha(\|x(t)\|)-\gamma(\|e(t^-)\|)) \le 0  \end{array} \right\}
\end{array}
\end{equation}
where $\theta \in \R_0^+$ is an additional design parameter.
Let us remark that the static ETM (\ref{eq:static}) can be seen as a limit case of the dynamic ETM (\ref{eq:dynamic}) when $\theta$ goes to $+\infty$.

In the following, we assume that for all $i\in \I$, $x(t_i)\ne 0$ (otherwise finite-time stabilization is obtained).
Let us remark that such an assumption is implicitly made in~\cite{tabuada2007} where stability of the closed loop system is proved by analyzing the evolution of the ratio ${\|e(t)\|}/{\|x(t)\|}$. The following lemma states that  $\eta$ remains non-negative for all time. 

\begin{lemma}
\label{lem:pos}
Let $\beta$ be a locally Lipschtiz continuous $\mathcal K_\infty$ function, $\sigma\in (0,1)$ and $\eta_0, \theta\in \R_0^+$,
let $x$, $e$, $\eta$ be given by (\ref{eq:clsys}), (\ref{eq:error}), (\ref{eq:s}) and (\ref{eq:dynamic}). Then,
for all $t\in [0,t_\infty)$,
$
\eta(t)+\theta(\sigma \alpha(\|x(t)\|)-\gamma(\|e(t)\|)) \ge 0 \text{ and } \eta(t)\ge 0.
$
\end{lemma}

\begin{proof} By construction, the ETM (\ref{eq:dynamic}) ensures that for all $t\in [0,t_\infty)$,
$$
\eta(t)+\theta(\sigma \alpha(\|x(t)\|)-\gamma(\|e(t^-)\|)) \ge 0. 
$$
By remarking that for all $t\in [0,t_\infty)$, $\|e(t)\| \le \|e(t^-)\|$, we obtain the first inequality. 
If $\theta=0$, the second inequality is equivalent to the first one. Then, let us assume that $\theta\ne 0$. 
The first inequality gives us 
$$
\sigma \alpha(\|x(t)\|)-\gamma(\|e(t)\|) \ge -\frac{1}{\theta} \eta(t).
$$
Then, from (\ref{eq:s}), we have that for all $t\in [0,t_\infty)$,
$$
\dot \eta(t) \ge -\beta(\eta(t))-\frac{1}{\theta} \eta(t), \; \eta(0)\ge 0.
$$
Then by the Comparison Lemma, it follows that $\eta(t)\ge 0$, for all $t\in [0,t_\infty)$.
\end{proof}

\subsubsection{Stability analysis}

To show the asymptotic stability of the closed loop system, we consider the following candidate Lyapunov function $W:\R^n\times \R_0^+: \rightarrow \R^+_0$ for the augmented dynamical system given by (\ref{eq:clsys}) and (\ref{eq:s}):
$$
W(x,\eta)=V(x)+\eta.
$$
It is clear that $W$ is positive definite and radially unbounded. 
Moreover, for all $(x,\eta)\in \R^n\times \R_0^+$, we have
$W(x,\eta) \ge V(x)$. Also, for all  $t\in [0,t_\infty)$, we have
\begin{eqnarray}
\nonumber
\frac{d}{dt}W(x(t),\eta(t))& \le &  -\alpha(\|x(t)\|) + \gamma(\|e(t)\|)   + \dot \eta(t) \\
\label{eq:diffW}
             & \le & (\sigma-1)\alpha(\|x(t)\|) - \beta(\eta(t))
\end{eqnarray}
which guarantees that $W(x(t),\eta(t))$ decreases and that $x(t)$ and $\eta(t)$ converge asymptotically to the origin provided that $t_\infty=+\infty$.

The following proposition shows that, for a given state of the system, the next execution time given by
a dynamic ETM is larger than that given by a static ETM. 
\begin{proposition}
\label{pro:time1}
Let $\beta$ be a locally Lipschtiz continuous $\mathcal K_\infty$ function, $\sigma\in (0,1)$ and $\eta_0, \theta\in \R_0^+$,
let $i\in \I$, $t_i\in \R_0^+$, $x(t_i) \in \R^n$ and $\eta(t_i)\ge 0$, let $t_{i+1}^s$ be given by the rule (\ref{eq:static}), let 
$t_{i+1}^d$ be given by the rule (\ref{eq:dynamic}), then $t_{i+1}^s\le t_{i+1}^d$.
\end{proposition}

\begin{proof} Let us assume that  $t_{i+1}^s > t_{i+1}^d$. Then, by (\ref{eq:static}), we must have
\begin{equation}
\label{eq:pro1}
\sigma \alpha(\|x(t_{i+1}^d)\|) - \gamma(\|e(t_{i+1}^{d-})\|)>0.
\end{equation}
We will now consider two different cases. If $\theta>0$, then we must have by (\ref{eq:dynamic}) and Lemma~\ref{lem:pos}
\begin{eqnarray*}
0 &\ge&  \eta(t_{i+1}^d)+\theta(\sigma \alpha(\|x(t_{i+1}^d)\|)-\gamma(\|e(t_{i+1}^{d-})\|)) \\
&\ge& \theta(\sigma \alpha(\|x(t_{i+1}^d)\|)-\gamma(\|e(t_{i+1}^{d-})\|))
\end{eqnarray*}
This contradicts (\ref{eq:pro1}). If $\theta=0$, then the triggering condition defined by (\ref{eq:dynamic}) gives
$\eta(t_{i+1}^d)=0$ and $\dot \eta(t_{i+1}^{d-}) \le 0$. Then, (\ref{eq:s}) gives
$$
0 \ge \dot \eta(t_{i+1}^{d-})= \sigma \alpha(\|x(t_{i+1}^d)\|)-\gamma(\|e(t_{i+1}^{d-})\|)
$$
which contradicts again (\ref{eq:pro1}). Hence, $t_{i+1}^s \le t_{i+1}^d$.
\end{proof}

\begin{remark}
\label{remark:theta}
The previous proposition has to be considered carefully, since it only shows that for a given state $x(t_i)$, the next execution time will be larger if we use a dynamic ETM (\ref{eq:dynamic}) rather than a static one (\ref{eq:static}). However, we cannot say anything on further execution times as generally we  have $x(t_{i+1}^s) \ne x(t_{i+1}^d)$ and thus we cannot apply the proposition again.
However, the proposition allows us to conclude that the minimum inter-execution time for the dynamic ETM (\ref{eq:dynamic}) cannot be smaller
than that for the static ETM  (\ref{eq:static}). 
Similarly, it can be shown that a smaller value of parameter $\theta$ results in a larger value of the minimum inter-execution time (the largest value being obtained for $\theta=0$).
\end{remark}

We can now state the following result on the stability of the closed loop system. 
\begin{theorem}\label{th:stab} Let us assume that $f$, $k$, $\alpha^{-1}$ and $\gamma$ are Lipschitz continuous on compacts.
Then, for all locally Lipschtiz continuous $\mathcal K_\infty$ functions $\beta$, $\sigma\in (0,1)$ and $\eta_0,\; \theta\in \R_0^+$, for all compact sets $S\subset \R^n$ containing the origin, there exists $\tau>0$ such that for all initial conditions $x(0)\in S$, the sequence $(t_i)_{i\in \I}$ defined by $(\ref{eq:dynamic})$ statisfies (\ref{eq:tau}). Moreover, $x(t)$ and $\eta(t)$ converge asymptotically to the origin.
\end{theorem}
\begin{proof} The theorem is essentially a consequence of Proposition~\ref{pro:time1} and of Theorem III.1 in~\cite{tabuada2007}.
Let $\mu=\max_{x\in S} V(x)+\eta_0$. Let $R$ be the compact set of all points $x\in \R^n$ such that $V(x)\le \mu$. 
In~\cite{tabuada2007}, it is shown that there exists $\tau>0$ such that if $x(t_i)\in R$ and if $t_{i+1}$ is generated according to rule (\ref{eq:static}),
then $t_{i+1}-t_i \ge \tau$.
Since $x(t_0)\in S$, we have that $W(x(t_0),\eta(t_0))=V(x(t_0))+\eta(t_0)\le \mu$.
Let us assume that for some $i\in \I$, $W(x(t_i),\eta(t_i))\le \mu$. Then, $V(x(t_i)) \le \mu$ and $x(t_i)\in R$. Let $t_{i+1}$ be 
generated according to rule (\ref{eq:dynamic}).
Then, it follows from Proposition~\ref{pro:time1} and Theorem III.1 in~\cite{tabuada2007} that $t_{i+1}-t_i \ge \tau$. Moreover, it follows from (\ref{eq:diffW}) that 
$W(x(t_{i+1}),\eta(t_{i+1}))\le \mu$. Hence, we have shown by induction that $t_{i+1}-t_i \ge \tau$, for all $i\in \I$.
Therefore, $t_\infty=+\infty$. Then,  (\ref{eq:diffW}) allows us to conclude that  $x(t)$ and $\eta(t)$ converge asymptotically to the origin.
\end{proof}

\subsubsection{Choice of parameters}
The proposed dynamic ETM has design parameters: a locally Lipschtiz continuous $\mathcal K_\infty$ function $\beta$ and reals $\sigma\in (0,1)$ and $\eta_0, \theta\in \R_0^+$.
In the following, we provide a discussion to give some insight on how to choose each parameter in order to tune the behavior of the system.
The following proposition shows how the decay rate of the Lyapunov function $V(x(t))$ can be tuned using parameters $\beta$, $\sigma$ and $\eta_0$:
\begin{proposition} 
\label{pro:param}
Let $\sigma \in (0,1)$, $\theta \in \R_0^+$, let $\eta_0= 0$ and let us assume that we can choose a locally Lipschtiz continuous $\beta$ such that 
\begin{equation}
\label{eq:choicebeta}
\forall r_1,r_2 \in \R_0^+,\;  \alpha\circ \overline\alpha^{-1} (r_1+r_2) \le \alpha \circ \overline\alpha^{-1} (r_1)+ \frac{1}{1-\sigma} \beta(r_2).
\end{equation}
If $\alpha \circ \overline\alpha^{-1}$ is locally Lipschitz continuous, then for all initial conditions $x(0)\in \R^n$, 
(\ref{eq:decaystatic}) holds.
\end{proposition}

\begin{proof}
From (\ref{eq:diffW}) and (\ref{eq:choicebeta}), it follows that for all $t\in [0,t_\infty)$
\begin{eqnarray*}
\frac{d}{dt}W(x(t),\eta(t))& \le&  (\sigma-1)\alpha\circ \overline \alpha^{-1}(V(x(t)))  - \beta(\eta(t))\\
& \le&  (\sigma-1)\alpha \circ \overline\alpha^{-1}(W(x(t),\eta(t))).
\end{eqnarray*}
This gives by the Comparison Lemma that   for all $t\in [0,t_\infty)$
\begin{eqnarray}
\nonumber
V(x(t)) & \le & W(x(t),\eta(t)) \le   \phi \left(W (x(0),\eta(0)),(1-\sigma)t \right) \\
\label{eq:decaydynamic2}
& \le & \phi \left(V (x(0)),(1-\sigma)t \right). 
\end{eqnarray}
\end{proof}

Hence, we can see that by choosing $\beta$ such that (\ref{eq:choicebeta}) holds and $\eta_0=0$,  the trajectories of the closed loop systems with static ETM and dynamic ETM have the same guaranteed decay rate given by (\ref{eq:decaystatic})  which can be tuned by choosing appropriately the parameter $\sigma$.
In addition, let us remark that if we $\alpha\circ \overline\alpha^{-1}$ is Lipschitz continuous on compacts, and we are only interested in the dynamics of the system for a compact subset of initial conditions $S\subset \R^n$, it is possible to choose a linear function $\beta$ such that (\ref{eq:choicebeta}) holds for all $r_1,r_2 \in [0,\mu]$ where $\mu=\max_{x\in S} V(x)$. 
Then, we can show that for all initial conditions $x(0)\in S$,
 (\ref{eq:decaystatic}) holds.

It should be noticed that, unlike when using the static ETM (\ref{eq:static}), the function $V(x(t))$ may not be a decreasing function.  
However, we know that it is upper-bounded by the function $W(x(t),\eta(t))$ which is decreasing  according to (\ref{eq:diffW}). 
Moreover, under the assumptions of Proposition \ref{pro:param},
the potential increase of the function $V(x(t))$ can be tuned using the parameter $\theta$. Indeed
if $\theta>0$, Lemma~\ref{lem:pos} and (\ref{eq:decaydynamic2}) give us that for all $t\in [0,t_\infty)$
\begin{eqnarray}
\nonumber
\frac{d}{dt}V(x(t))& \le& -\alpha(\|x(t)\|)+\gamma(\|e(t)\|) \\
\nonumber &\le & (\sigma-1)\alpha(\|x(t)\|) + \frac{1}{\theta} \eta(t) \\
\nonumber
& \le & (\sigma-1)\alpha(\|x(t)\|) \\
\nonumber
&& +\frac{1}{\theta} \left(  \phi \left(V (x(0)),(1-\sigma)t \right) - V(x(t)) \right).
\end{eqnarray}
It appears that $V(x(t))$ can increase only when its value is far from the prescribed decay rate given by (\ref{eq:decaystatic}).
Moreover, the larger the value of $\theta$ the more limited the increase. 
We have seen that the parameters of the dynamic ETM
can be chosen in order to tune the dynamical properties of the function $V(x(t))$.
These parameters also have an influence on the inter-execution times. We have already pointed out in
Remark \ref{remark:theta}, that smaller values of $\theta$ yield a larger minimum inter-execution time: this parameter allows us to adjust the balance between the potential increase of the function $V(x(t))$ and the value of the minimum inter-execution time.

\section{The Case of Linear Systems}
\label{sec:3}

In this section, we specify the results developed above to  linear systems of the form
\begin{equation}
\label{eq:lin}
\dot x= Ax+Bu,\; x\in \R^n, \; u\in \R^m.
\end{equation}
We assume that we are given a linear feedback controller $u=Kx$ such that the ``ideal'' closed loop system
\begin{equation}
\label{eq:clin}
\dot x= Ax+BKx,\; x\in \R^n, 
\end{equation}
is globally asymptotically stable.
This implies the existence of a Lyapunov function $V(x)=x^\top P x$ where $P$ is a symmetric positive definite matrix such that 
\begin{equation}
\label{eq:lyaplin}
(A+BK)^\top P + P (A+BK) = -Q
\end{equation}
where $Q$ is an arbitrary symmetric positive definite matrix. Then, there exist $\chi \ge \kappa>0$ such that $\chi P \ge Q\ge \kappa P$.
When the controller is implemented on a digital platform we have with the notations of the previous section, the dynamics of the closed loop system that is given by
\begin{equation}
\label{eq:clin1}
\dot x= Ax+BK(x+e).
\end{equation}
Then, for all $ t\in [0, t_\infty)$,
\begin{equation}
\label{eq:dlyaplin}
\frac{d}{dt}V(x(t)) = -x(t)^\top Q x(t) + 2 x(t)^\top P BK e(t).
\end{equation}

\subsection{Event triggering mechanisms}

For an event-triggered implementation, one may use, as suggested for instance in~\cite{HJT2012}, the static ETM defined by the following rule:
\begin{equation}
\label{eq:staticlin}
\begin{array}{l}
t_0=0,\\
t_{i+1} = \inf\left\{\begin{array}{l} t\in \R |\; t>t_i \land \\ 
\sigma x(t)^\top Q x(t)-2 x(t)^\top P BK e(t^-) \le 0 \end{array} \right\}
\end{array}
\end{equation}
where $\sigma\in (0,1)$ is a design parameter. Adapting the proof of~\cite{tabuada2007} (see also~\cite{HJT2012}), it is possible to show that 
for all $\sigma \in (0,1)$, 
there exists 
$\tau >0 $ such that  (\ref{eq:tau}) holds. Moreover, it follows that 
\begin{equation}
\label{eq:decaylin}
V(x(t)) \le V(x(0)) e^{(\sigma-1)\kappa t},\;  \forall t\in [0, t_\infty).
\end{equation}
Hence, by choosing $\sigma$ close to $0$, the decay rate of the Lyapunov function $V(x(t))$
approaches that of the ``ideal'' system (\ref{eq:clin}).

We now propose to use a dynamic ETM based on the internal dynamic variable $\eta$ satisfying the following differential equation:
\begin{equation}
\label{eq:s2}
\dot \eta = - \lambda\eta + \sigma x^\top Q x - 2 x^\top P BK e, \; \eta(0)=\eta_0
\end{equation}
where $\sigma\in (0,1)$, $\lambda>0$ and $\eta_0 \in \R_0^+$ are design parameters. 
Let us remark that $\eta$ is just a filtered value of the signal $ \sigma x^\top Q x - 2 x^\top P BK e$.
The dynamic ETM is then defined by the following rule:
\begin{equation}
\label{eq:dynamiclin}
\begin{array}{l}
t_0=0,\\
t_{i+1} = \inf\left\{\begin{array}{l} t\in \R |\; t>t_i \land \eta(t)+\\ 
\theta(\sigma x(t)^\top Q x(t)-2 x(t)^\top P BK e(t^-)) \le 0 \end{array} \right\}
\end{array}
\end{equation}
where $\theta\in \R_0^+$ is an additional design parameter.
The static ETM (\ref{eq:staticlin}) can be seen as a limit case of the dynamic ETM (\ref{eq:dynamiclin}) when $\theta$ goes to $+\infty$.

\begin{remark}
\label{remark:dyn}
For the dynamic ETM (\ref{eq:dynamiclin}),  it is possible to show by adapting the proof of Lemma~\ref{lem:pos} 
that for all $t\in [0,t_\infty)$,
$\eta(t)+\theta(\sigma x(t)^\top Q x(t)-2 x(t)^\top P BK e(t)) \ge 0 \text{ and } \eta(t) \ge 0$.
Also, a similar result to Proposition~\ref{pro:time1} can be shown stating that the next execution time is larger when using the dynamic ETM (\ref{eq:dynamiclin}) rather than the static ETM (\ref{eq:staticlin}). Similarly, it can be shown that a smaller value of parameter $\theta$ results in a larger value of the next execution time.
\end{remark}

Similar to Theorem \ref{th:stab}, we can state the following result on the stability of the closed loop system:
\begin{theorem} For all $\lambda>0$, $\sigma\in (0,1)$ and $\eta_0,\theta\in \R_0^+$, there exists  $\tau>0$ such that for all initial conditions $x(0)\in \R^n$, the sequence $(t_i)_{i\in \I}$ defined by $(\ref{eq:dynamiclin})$ statisfies (\ref{eq:tau}). Moreover, $x(t)$ and $\eta(t)$ converge asymptotically to the origin. 
\end{theorem}

\begin{proof} It is shown in~\cite{tabuada2007} and~\cite{HJT2012} that there exists $\tau>0$ such that for all $i\in\I$, for any value of $x(t_i)\in \R^n$,
the execution time $t_{i+1}$ generated using the rule
(\ref{eq:staticlin}) satisfies $t_{i+1}-t_i>\tau$. Then, it follows from Remark~\ref{remark:dyn} that if $t_{i+1}$ is generated using the rule 
(\ref{eq:dynamiclin}) then we will also have $t_{i+1}-t_i>\tau$. This also implies that $t_\infty=+\infty$.
Now let us consider the candidate Lyapunov function $W:\R^n\times \R_0^+\rightarrow \R_0^+$, for the augmented system given by (\ref{eq:clin1}) and
(\ref{eq:s2}), defined by $W(x,\eta)=V(x)+\eta$. Then,
for all $t\in \R_0^+$,
\begin{equation}
\label{eq:diffWlin}
\frac{d}{dt}W(x(t),\eta(t)) = (\sigma-1)x(t)^\top Q x(t) - \lambda\eta(t)
\end{equation}
which shows that $W(x(t),\eta(t))$ decreases and that $x(t)$ and $\eta(t)$ converge asymptotically to the origin.
\end{proof}

\subsection{Influence of parameters}

The proposed dynamic ETM has several design parameters  $\lambda>0$, $\sigma\in (0,1)$ and $\eta_0,\theta\in \R_0^+$.
In the following, we present a certain number of results that can guide us in choosing values for these parameters. 
The proofs can be found in appendix.

\subsubsection{Minimum inter-execution time}
We first establish a lower bound on the minimum inter-execution time.
\begin{proposition}
\label{pro:tau}
Let $\lambda>0$, $\sigma\in (0,1)$ and $\eta_0,\theta\in \R_0^+$. Then
for all initial conditions $x(0)\in \R^n$, the sequence $(t_i)_{i\in \I}$ defined by $(\ref{eq:dynamiclin})$ satisfies (\ref{eq:tau})
where $\tau>0$ is given by
\begin{enumerate}
\item If $a \le \lambda/2$, 
\begin{equation}
\label{eq:tau1}
\tau = \int_0^1 \frac{1}{a\frac{p}{\sigma q}+(a+b)s+b\frac{\sigma q}{p}s^2} ds
\end{equation}
\item If $a>\lambda/2$, and $\theta\le 1/(2a-\lambda)$,
\begin{equation}
\label{eq:tau2}
\tau = \int_0^1 \frac{1}{a\frac{p}{\sigma q}+(a+b)s+b\frac{\sigma q}{p}s^2+(a-\frac{\lambda}{2})(s^3-s)} ds
\end{equation}
\item If $a>\lambda/2$, and $\theta > 1/(2a-\lambda)$,
\begin{equation}
\label{eq:tau3}
\tau = \int_0^1 \frac{1}{a\frac{p}{\sigma q}+(a+b)s+b\frac{\sigma q}{p}s^2+\frac{1}{2\theta}(s^3-s)} ds
\end{equation}
\end{enumerate} 
 with 
$$
q=\lambda_{\min}(Q),\; p=2\|PBK\|, \; a=\|A+BK\|, \; b=\|BK\|.
$$
\end{proposition}

For the first case, the lower bound (\ref{eq:tau1}) coincides with the lower bound
 on minimum inter-execution times for static ETM given in \cite{tabuada2007}.
We give an intuitive interpretation as follows, if $\lambda/2 \ge a$ then the filter (\ref{eq:s2}) is too fast (time constant $1/\lambda$)
in comparison to filtered signal $\sigma x^\top Q x - 2 x^\top P BK e$ (time constant approximated by $1/(2a)$).
Then, there is essentially no gain  filtering the signal and the dynamic ETM does not guarantee a larger minimum inter-execution time than 
the static ETM.
For the second and third case, it should be noted that lower bounds (\ref{eq:tau2}) and (\ref{eq:tau3}) are strictly larger than that obtained for the static ETM.
It can be seen that our lower bound is a continuous function of $\theta$, constant on $[0,1/(2a-\lambda)]$, and strictly decreasing 
on  $[1/(2a-\lambda),+\infty)$. When $\theta$ goes to $+\infty$ the lower bound given by (\ref{eq:tau3}) tends toward that obtained for the static ETM,
 which is consistent with the fact that the static ETM can be seen as the limit case of the dynamic ETM when $\theta$ goes to $+\infty$.

\subsubsection{Decay rate}

Similar to the case of nonlinear systems, we can show that a suitable choice of parameters
allows us to guarantee the same decay rate of the function $V(x(t))$ as with the static ETM (\ref{eq:staticlin}).
\begin{proposition}
\label{pro:param1}
Let $\sigma\in (0,1)$, $\theta\in \R_0^+$, $\eta_0=0$ and $\lambda=(1-\sigma)\kappa$.
Then,
for all initial conditions $x(0)\in \R^n$, (\ref{eq:decaylin}) holds.
\end{proposition}

Let us remark that the decay rate  can be tuned by choosing appropriately the parameter $\sigma$.
In the previous proposition the value of parameter $\theta$ does not affect the decay rate of the function $V(x(t))$.
To understand the influence of parameter $\theta$ one has to consider a quadratic integral performance index.

\subsubsection{Quadratic integral performance index}

We now examine the influence of parameters on the performance of the closed loop system with respect to the following quadratic integral cost:
$$
J(x(0)) = \int_0^{+\infty} x^\top(t) Q x(t) dt
$$
Such performance criteria have been considered in the context of event-triggered control in~\cite{antunes2012dynamic}.
In the case of the ``ideal'' closed loop system (\ref{eq:clin}), one would have
$J(x(0))=V(x(0))$. 
Using the dynamic ETM, we have the following guarantees:
\begin{proposition}
\label{pro:param2}
Let $\sigma\in (0,1)$, $\theta\in \R_0^+$, $\eta_0=0$ and $\lambda=(1-\sigma)\kappa$.
Then,
for all initial conditions $x(0)\in \R^n$,
\begin{equation}
\label{eq:perfdynamiclin}
J(x(0))  \le   \frac{V (x(0))}{1-\sigma}  \frac{1/\kappa + \theta(1-\sigma)}{1/\chi + \theta (1-\sigma)}
\end{equation}
\end{proposition}

The previous proposition shows that the bound on the proposed performance criteria is increasing with $\theta$.
Hence, the dynamic ETM (\ref{eq:dynamiclin}) degrades the performance when compared to the static ETM (\ref{eq:staticlin})
whose performance are approached when $\theta \rightarrow +\infty$. By tuning $\sigma$ and $\theta$, one can approach arbitrarily close 
the performance of the ``ideal'' closed loop system (\ref{eq:clin}).

\subsubsection{A possible choice of parameters}
\label{choice}

Propositions \ref{pro:tau}, \ref{pro:param1} and \ref{pro:param2} suggest a strategy for the choice of parameters.
Firstly, we tune the decay rate of the Lyapunov function according to Proposition \ref{pro:param1} by choosing
$\eta_0=0$ and $\lambda=(1-\sigma)\kappa$; the value  of $\sigma$ determines the degradation of 
the decay rate of the Lyapunov function with respect to the ``ideal'' closed loop system.

Then, it remains to choose the parameter $\theta\in  \R_0^+$.
Let us remark that $a$ is an upper bound for the spectral radius $\rho(A+BK)$.
Then (\ref{eq:lyaplin}) implies that we have $\kappa \le  2\rho(A+BK) \le 2 a$. Hence, $\lambda =(1-\sigma)\kappa \le 2a$
and we are in either in the second or in the third case of Proposition \ref{pro:tau}.
Then, the best lower bound on the minimum inter-execution time is obtained for $\theta \in [0, 1/(2a-\lambda)]$.
Proposition \ref{pro:param2} suggests to choose $\theta$ as large as possible in order to minimize the degredation 
of the  performance index.
Thus, it seems reasonable to choose $\theta=1/(2a-\lambda)$.

\section{Example}
\label{sec:example}

To evaluate our approach, we consider the example introduced in~\cite{tabuada2007} of the form (\ref{eq:lin}) with
$
A=
\left[
\begin{smallmatrix}
0 & 1 \\
-2 & 3
\end{smallmatrix}
\right], \;
B=
\left[
\begin{smallmatrix}
0  \\
1
\end{smallmatrix}
\right].
$
A stabilizing controller is given by the control gain $K=\left[ \begin{smallmatrix} 1 & -4 \end{smallmatrix} \right]$ with an associated Lyapunov function
$V(x)=x^\top P x$ satisfying (\ref{eq:lyaplin}) with
$
P=
\left[
\begin{smallmatrix}
1 & 0.25 \\
0.25 & 1
\end{smallmatrix}
\right], \;
Q=
\left[
\begin{smallmatrix}
0.5 & 0.25 \\
0.25 & 1.5
\end{smallmatrix}
\right].
$
We have implemented the event-triggered control schemes given by the static and dynamic ETM (\ref{eq:staticlin}) and
(\ref{eq:dynamiclin}). We have used several values of parameter $\sigma$, for the dynamic ETM, we 
chose $\eta_0=0$, $\lambda=(1-\sigma)\kappa $ with $\kappa=0.48$, and experimented several values of $\theta$. 
We have also implemented the following ETM taken from~\cite{mazo2010}:
\begin{equation}
\label{eq:etm}
\begin{array}{l}
t_0=0,\\
t_{i+1} = \inf\left\{\begin{array}{l} t\in \R |\; t>t_i \land \\ 
V(x(t))\ge e^{(\sigma-1)\kappa(t-t_i)}V(x(t_i))  \end{array} \right\}
\end{array}
\end{equation}
It can be shown that using this ETM ensures that the function $V(x(t))$ satisfies (\ref{eq:decaylin}) and that
the quadratic integral cost $J(x(0))$ satisfies (\ref{eq:perfdynamiclin}) with $\theta=0$.
Also, it is not hard to show that the minimum inter-execution time for this ETM is always larger than that of (\ref{eq:staticlin}) and
(\ref{eq:dynamiclin}). For a given value of $\sigma$, the closed loop systems with these three ETM have the same guarantees on the decay rate and are thus comparable.

On Figure~\ref{fig}, we have represented, for $\sigma=0.1$ and initial condition 
$x(0)=\left[\; 10 \; \; 0\; \right]^\top$,  the evolution of the functions $V(x(t)$ and $W(x(t),\eta(t))$ using a static ETM (\ref{eq:staticlin}), 
dynamic ETM (\ref{eq:dynamiclin}) with $\theta=1/(2a-\lambda)$ and ETM (\ref{eq:etm}).
It can be seen that the inter-execution time is significantly larger using the   ETM (\ref{eq:dynamiclin}) and  (\ref{eq:etm}).
Also, it can be seen that ETM (\ref{eq:etm}) produces large variations of the function $V(x(t))$ while these variations are much smaller for the dynamic ETM (\ref{eq:dynamiclin}), when
$V(x(t))$ is monotonically decreasing for the static ETM (\ref{eq:staticlin}). 
Hence, it seems that ETM (\ref{eq:dynamiclin}) achieves a good compromise maximizing the inter-execution time while minimizing the variations of $V(x(t))$.

\begin{figure}[!h]
\begin{tabular}{ccc}
\hspace{-0.6cm}
\includegraphics[width=0.37\columnwidth]{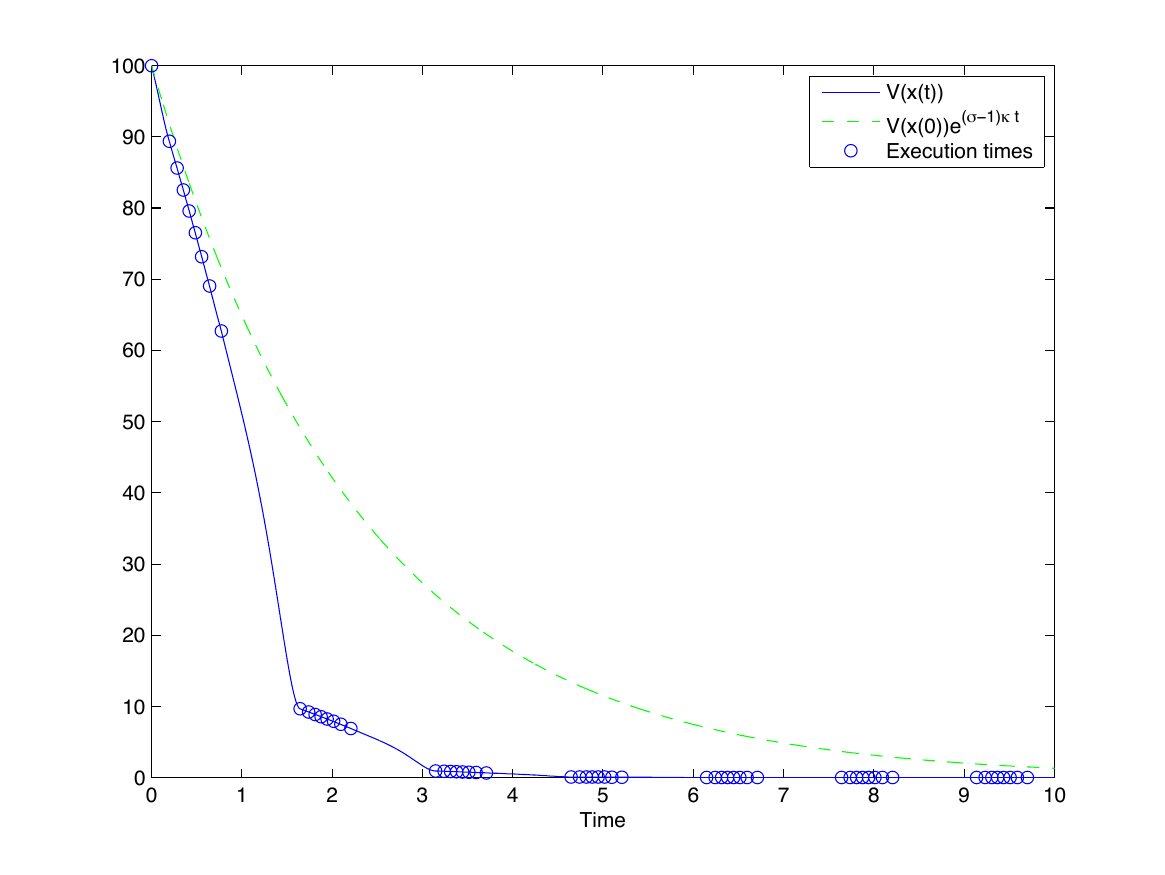} & 
\hspace{-0.8cm}
\includegraphics[width=0.37\columnwidth]{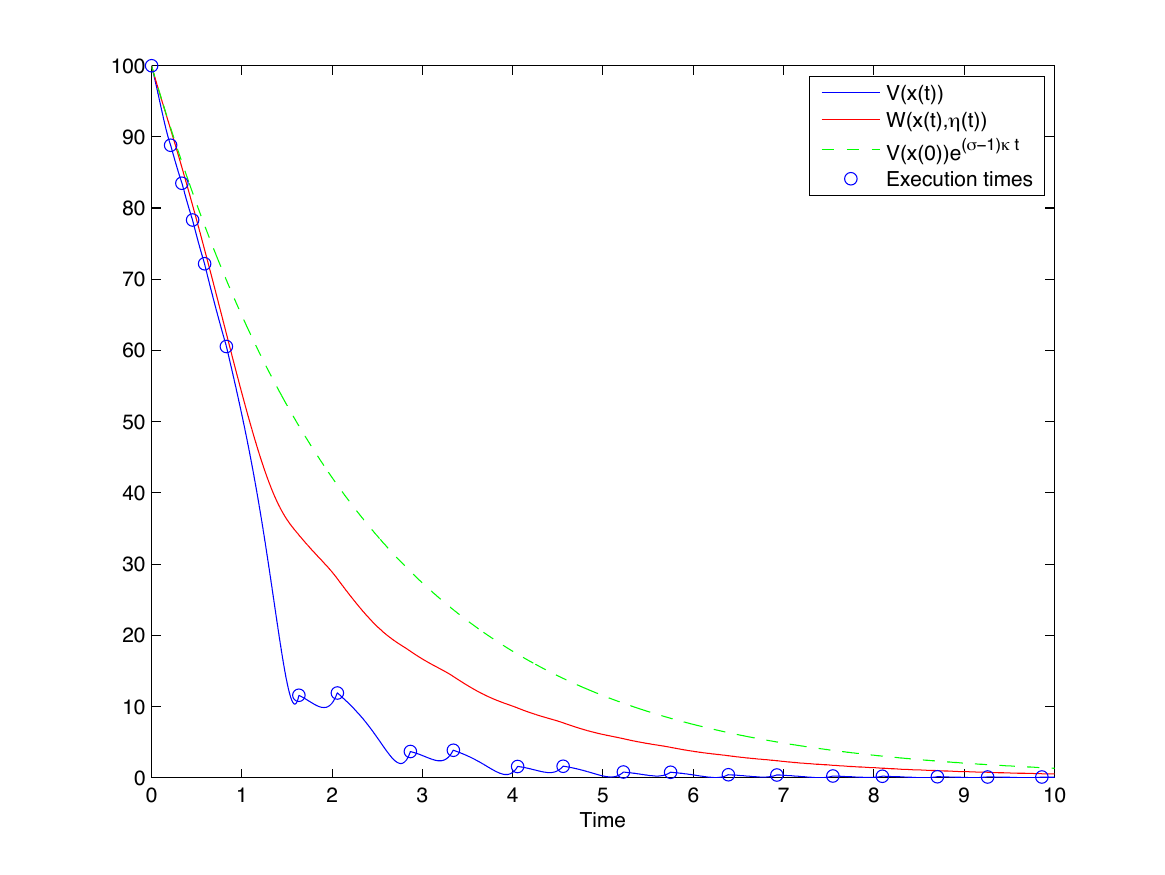}& 
\hspace{-0.8cm}
\includegraphics[width=0.37\columnwidth]{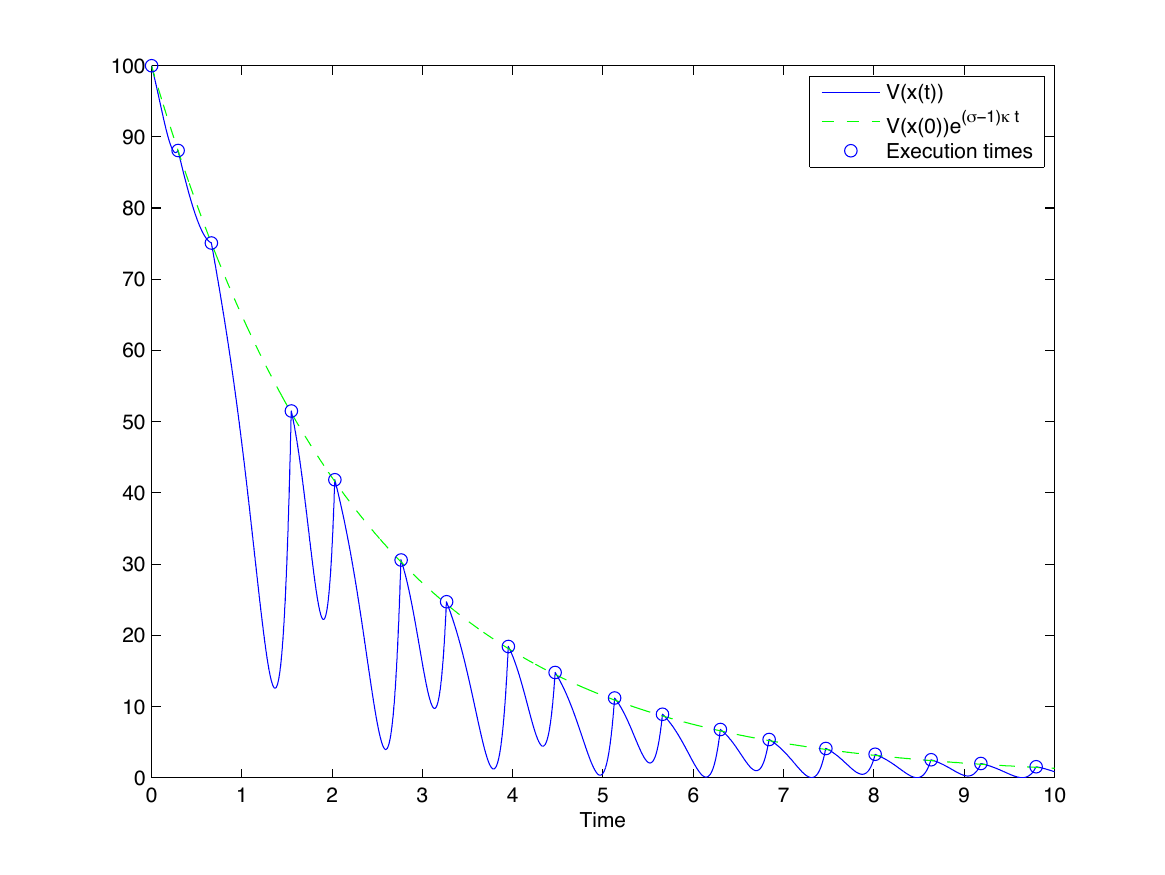}\\
\hspace{-0.6cm} {\footnotesize Static ETM (\ref{eq:staticlin}) }&\hspace{-0.8cm} {\footnotesize Dynamic ETM (\ref{eq:dynamiclin})} &\hspace{-0.8cm} {\footnotesize  ETM (\ref{eq:etm})}
\end{tabular}
\caption{Evolution of the functions $V(x(t)$ and $W(x(t),\eta(t))$ using ETM (\ref{eq:staticlin}),  (\ref{eq:dynamiclin}) and (\ref{eq:etm}) for $\sigma=0.1$, $\lambda=(1-\sigma)\kappa$,
$\eta_0=0$, $\theta=1/(2a-\lambda)$ and initial condition 
$x(0)=\left[\; 10 \; \; 0\; \right]^\top$.}  
\label{fig}                                 
\end{figure}

%

\begin{table}[!t]
\begin{center}
\begin{footnotesize}
\begin{tabular}{|l|c|c|c|}
\multicolumn{4}{c}{{\bf Mean value of inter-execution time}} \\
\hline
                     & $\sigma=0.001$ & $\sigma=0.01$ & $\sigma=0.1$ \\
\hline
ETM (\ref{eq:staticlin})               & 0.003s  &  0.025s  &  0.178s \\
\hline
ETM (\ref{eq:dynamiclin}) ($\theta=0$)         & 0.127s  &  0.452s  &  0.581s \\
ETM (\ref{eq:dynamiclin}) ($\theta=0.01$) & 0.144s  &  0.470s  &  0.579s \\
ETM (\ref{eq:dynamiclin}) ($\theta=1$)    & 0.152s  &  0.410s  &  0.551s \\
ETM (\ref{eq:dynamiclin}) ($\theta=100$)   & 0.068s  &  0.202s  &  0.424s \\
\hline
ETM (\ref{eq:etm}) & 0.588s & 0.591s & 0.590s \\
\hline
\multicolumn{4}{c}{} \\
\multicolumn{4}{c}{{\bf Coefficient of variation of inter-execution time} } \\
\hline
                     & $\sigma=0.001$ & $\sigma=0.01$ & $\sigma=0.1$ \\
\hline
ETM (\ref{eq:staticlin})                & 49.38 &   5.92  &  0.86 \\
\hline
ETM (\ref{eq:dynamiclin}) ($\theta=0$)  & 0.32  &  0.09 &   0.07 \\
ETM (\ref{eq:dynamiclin}) ($\theta=0.01$) & 0.37 &    0.11 &     0.09 \\
ETM (\ref{eq:dynamiclin}) ($\theta=1$) & 0.92 &   0.34 &   0.26 \\
ETM (\ref{eq:dynamiclin}) ($\theta=100$)   & 2.21 &   0.74 &   0.36 \\
\hline
ETM (\ref{eq:etm}) & 0.12 &   0.15 &   0.14  \\
\hline
\multicolumn{4}{c}{} \\
\multicolumn{4}{c}{{\bf Mean value of performance index $\Delta J(x(0))$}} \\
\hline
                     & $\sigma=0.001$ & $\sigma=0.01$ & $\sigma=0.1$ \\
\hline
ETM (\ref{eq:staticlin})                & 0.54 &    0.54 &  0.55 \\
\hline
ETM (\ref{eq:dynamiclin}) ($\theta=0$)  & 0.87   & 0.87  &  0.94 \\
ETM (\ref{eq:dynamiclin}) ($\theta=0.01$) & 0.83 &  0.84 & 0.91 \\
ETM (\ref{eq:dynamiclin}) ($\theta=1$) & 0.57 &   0.57 &    0.58 \\
ETM (\ref{eq:dynamiclin}) ($\theta=100$)   & 0.54 &  0.54 & 0.55 \\
\hline
ETM (\ref{eq:etm}) & 1.44 &  1.46 &  1.58 \\
\hline
\multicolumn{4}{c}{} \\
\end{tabular}
\end{footnotesize}
\vspace{-0cm}
\caption{Mean value, variability of inter-execution times, and mean value of performance index.}  
\label{tab}                                 
\end{center}                                 
\end{table}

\begin{figure}[!t]
\vspace{-0.5cm}
\begin{center}
\includegraphics[width=0.6\columnwidth]{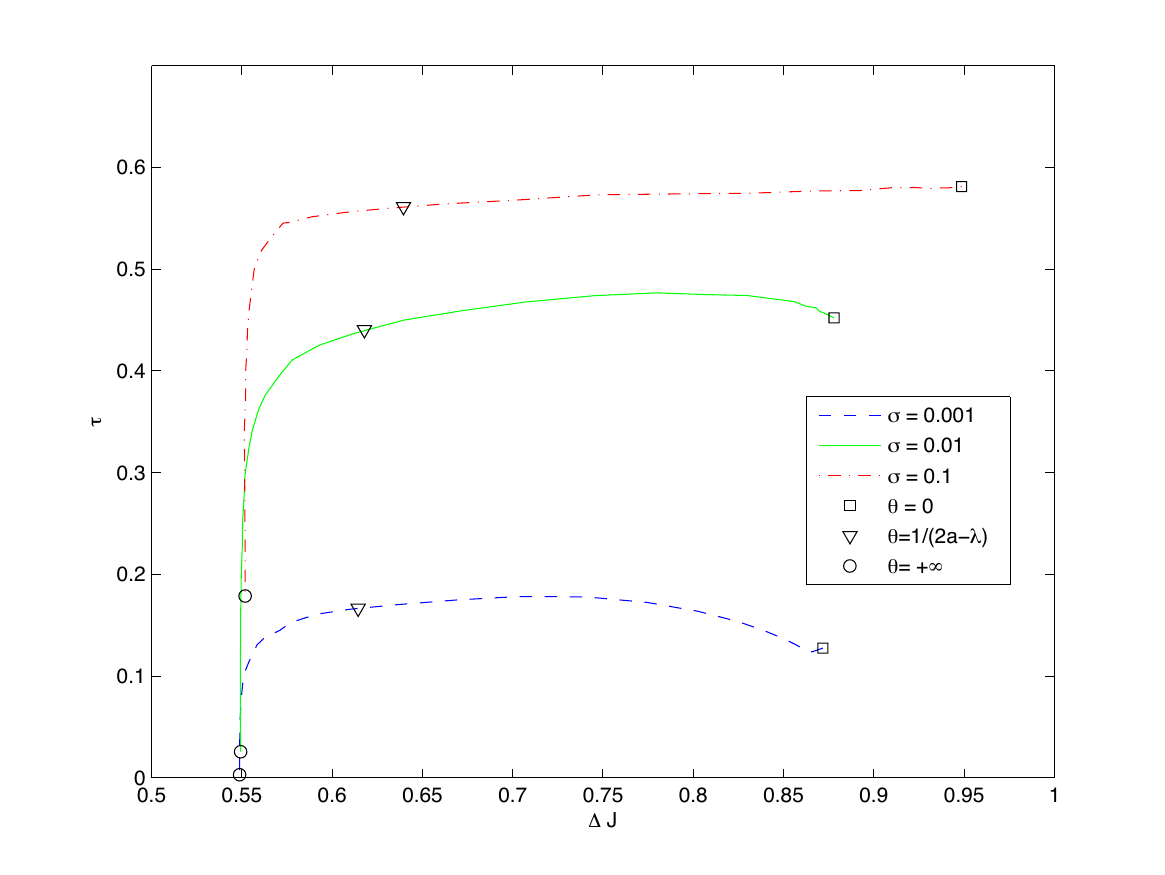} 
\caption{Mean values of the inter-execution time and of the performance index obtained using
the dynamic ETM (\ref{eq:dynamiclin}) for several values of parameters $\sigma$ and $\theta$ with $\lambda=(1-\sigma)\kappa$,
$\eta_0=0$.}  
\label{fig1}                                 
\end{center}    
\vspace{-0.5cm}                             
\end{figure}

We also ran simulations for several parameter values and initial values given by
$
x(0)=\left[\begin{smallmatrix} 10\cos\left(\frac{2\pi}{60}i\right) & 10\sin\left(\frac{2\pi}{60}i\right) \end{smallmatrix} \right]^\top, \;
i=1,\dots 60.
$
We ran the simulations on a frame of $10$s. The mean values and the coefficients of variability (ratio between the standard deviation and the mean value) of  inter-execution times
as well as the mean value of the normalized performance index $\Delta J(x(0))=J(x(0))/V(x(0))$ are reported in Table~\ref{tab}. 
It can be seen that the use of dynamic ETM (\ref{eq:dynamiclin})
results in significantly larger inter-execution times in average than the static ETM. The gain is considerable for small values of $\sigma$. For $\sigma=0.1$,
it allows us to increase the inter-execution times by a factor between $2$ and $3$.
It is noticeable that in all cases the largest inter-execution time is achieved 
by the ETM (\ref{eq:etm}) and that the parameter $\sigma$ has little influence here.
Also, looking at the coefficient of variability, it appears that 
the use of 
ETM (\ref{eq:dynamiclin}) or (\ref{eq:etm}) reduces significantly the variability of the inter-execution times and thus renders the behavior of the system more predictable. 
The performance index $\Delta J(x(0))$ is quite similar for the dynamic ETM (\ref{eq:dynamiclin}) with larger value of $\theta$ and the static ETM (\ref{eq:staticlin}), while it is not so good for the ETM (\ref{eq:etm}). 
Hence, it appears that the main advantage of the dynamic ETM (\ref{eq:dynamic}) is that it allows to find a compromise between a the inter-execution times and the performance index. This is corroborated by Figure~\ref{fig1}, where we represented the mean values of the  inter-execution time and of the performance index obtained using
the dynamic ETM (\ref{eq:dynamiclin}).
Each curve corresponds to a given value of $\sigma$ with parameter $\theta$ ranging from $0$ to $+\infty$ (which corresponds to the static ETM (\ref{eq:staticlin})).
The three curves are roughly made of two branches. 
On the first branch (starting from $\theta=0$), the inter-execution time remains roughly constant while the performance index improves as $\theta$ increases.
On the second branch (ending at $\theta=+\infty$), the performance index remains constant while the inter-execution time gets smaller as $\theta$ increases.
Then, it seems that there is an optimal value for parameter $\theta$ which corresponds to the intersection of the two branches.
Let us remark that the value $\theta=1/(2a-\lambda)$ suggested in section \ref{choice} is quite close to the optimal value.

\section{Conclusion}

In this paper, we have presented a new class of dynamic ETM for event-triggered control systems.  For nonlinear systems,
we have proved the stability of the resulting closed loop system. 
Further results have been shown for linear systems which give some insight on how the parameters of the ETM can be chosen.
This paper  has several potential applications. Indeed, similar schemes could certainly be used in the contexts of decentralized~\cite{mazo2011}, output-based~\cite{donkers2012}, or periodic~\cite{heemels2013} event-triggered control systems.
Also, dynamic ETM could be used to derive new algorithms for self-triggered control systems~\cite{wang2009,anta2010,mazo2010}.

\subsubsection*{Acknowledgements} The author would like to thank Romain Postoyan for valuable comments on a preliminary version of this paper.

\bibliographystyle{plain}
\bibliography{event}

\begin{thebibliography}{10}

\bibitem{anta2010}
A.~Anta and P.~Tabuada.
\newblock To sample or not to sample: self-triggered control for nonlinear
  systems.
\newblock {\em IEEE Transactions on Automatic Control}, 55(9):2030--2042, 2010.

\bibitem{antunes2012dynamic}
D.~Antunes, W.P.M.H. Heemels, and P.~Tabuada.
\newblock Dynamic programming formulation of periodic event-triggered control:
  Performance guarantees and co-design.
\newblock In {\em IEEE Conference on Decision and Control}, pages 7212--7217,
  2012.

\bibitem{donkers2012}
M.C.F. Donkers and W.P.M.H. Heemels.
\newblock Output-based event-triggered control with guaranteed
  ${L}_\infty$-gain and improved and decentralised event-triggering.
\newblock {\em IEEE Transactions on Automatic Control}, 57(6):1362--1376, 2012.

\bibitem{heemels2013}
W.P.M.H. Heemels, M.C.F. Donkers, and A.R. Teel.
\newblock Periodic event-triggered control for linear systems.
\newblock {\em IEEE Transactions on Automatic Control}, 2013.
\newblock To appear.

\bibitem{HJT2012}
W.P.M.H. Heemels, K.H. Johansson, and P.~Tabuada.
\newblock An introduction to event-triggered and self-triggered control.
\newblock In {\em IEEE Conference on Decision and Control}, 2012.

\bibitem{khalil}
H.K. Khalil.
\newblock {\em Nonlinear systems}.
\newblock Prentice Hall, third edition edition, 2002.

\bibitem{mazo2010}
M.~{Mazo Jr}., A.~Anta, and P.~Tabuada.
\newblock An {ISS} self-triggered implementation for linear controllers.
\newblock {\em Automatica}, 46(8):1310--1314, 2010.

\bibitem{mazo2011}
M.~{Mazo Jr.} and P.~Tabuada.
\newblock Decentralized event-triggered control over wireless sensor/actuator
  networks.
\newblock {\em IEEE Transactions on Automatic Control}, 56(10):2456--2461,
  2011.

\bibitem{postoyan2011}
R.~Postoyan, A.~Anta, D.~Nesic, and P.~Tabuada.
\newblock A unifying {Lyapunov}-based framework for the event triggered control
  of nonlinear systems.
\newblock In {\em Joint IEEE Conference on Decision and Control and European
  Control Conference}, pages 2559--2564, 2011.

\bibitem{postoyan2011a}
R.~Postoyan, P.~Tabuada, D.~Nesic, and A.~Anta.
\newblock Event-triggered and self-triggered stabilization of distributed
  networked control systems.
\newblock In {\em Joint IEEE Conference on Decision and Control and European
  Control Conference}, pages 2565--2570, 2011.

\bibitem{seuret2011}
A.~Seuret and C.~Prieur.
\newblock Event-based sampling algorithms based on a {Lyapunov} function.
\newblock In {\em Joint IEEE Conference on Decision and Control and European
  Control Conference}, pages 6128--6133, 2011.

\bibitem{tabuada2007}
P.~Tabuada.
\newblock Event-triggered real-time scheduling of stabilizing control tasks.
\newblock {\em IEEE Transactions on Automatic Control}, 52(9):1680--1685, 2007.

\bibitem{wang2008}
X.~Wang and M.D. Lemmon.
\newblock Event design in event-triggered feedback control systems.
\newblock In {\em IEEE Conference on Decision and Control}, pages 2105--2110,
  2008.

\bibitem{wang2009}
X.~Wang and M.D. Lemmon.
\newblock Self-triggered feedback control systems with finite-gain ${L}_2$
  stability.
\newblock {\em IEEE Transactions on Automatic Control}, 45(3):452–467, 2009.

\end{thebibliography}

\section*{Appendix}

\subsection{Proof of Proposition \ref{pro:tau}}

\begin{proof}
Let us start by remarking that, 
\begin{eqnarray*}
\sigma x^\top Q x-2 x^\top P BK e & \ge & \sigma q \|x\|^2 - p \|x\| \|e\| \\
& \ge & \sigma q \|x\|^2 - \frac{p}{2} \left(\frac{\sigma q}{p} \|x\|^2 + \frac{p}{\sigma q}  \|e\|^2 \right) \\
& \ge & \frac{\sigma q}{2} \|x\|^2 -  \frac{p^2}{2\sigma q}  \|e\|^2 = q'  \|x\|^2 -  p'  \|e\|^2
\end{eqnarray*}
where $q'={\sigma q}/{2}$ and $p'={p^2}/{2\sigma q}$. 
Let us also remark that a lower bound the inter-execution time using the static ETM (\ref{eq:staticlin})
is given by the time it takes for the function $\frac{\sqrt{p'}\|e\|}{\sqrt{q'}\|x\|}=\frac{p\|e\|}{\sigma q\|x\|}$ to go from $0$ to $1$.
In~\cite{tabuada2007}, it is shown that this time is at least
$$
\int_0^1 \frac{1}{a\frac{p}{\sigma q}+(a+b)s+b\frac{\sigma q}{p}s^2} ds.
$$
Then, following the discussion in Remark~\ref{remark:dyn}, it follows that this also a valid lower bound for the inter-execution
time using the dynamic ETM (\ref{eq:dynamiclin}).
Hence, the first case of the proposition is proved and we can assume in the following that $a>\lambda/2$.
Let us assume for the moment that $\theta>0$.
We have
$$
 \eta+\theta(\sigma x^\top Q x-2 x^\top P BK e) \ge \eta +  \theta(q'  \|x\|^2 -  p'  \|e\|^2).
$$
Hence, it follows that a lower bound on the inter-execution time is given by the time it takes for the function
$$
\psi = \frac{\sqrt{\theta p'} \|e\|}{\sqrt{\eta+\theta q' \|x\|^2}}  
$$
to go from $0$ to $1$. 
We have
\begin{eqnarray*}
\dot \psi & = & \frac{\sqrt{\theta p'} e^\top \dot e}{\|e\| \sqrt{\eta+\theta q' \|x\|^2}} - \frac{\sqrt{\theta p'} \|e\|}{2(\eta+\theta q' \|x\|^2)^{3/2}} (\dot \eta + 2 \theta q' x^\top \dot x)
\end{eqnarray*}
By remarking that $\dot e=-\dot x$,  $\|\dot x \| \le a \|x\| + b \|e\|$
and $\dot \eta \ge -\lambda \eta +q'  \|x\|^2 -  p'  \|e\|^2$, it follows that
\begin{eqnarray}
\nonumber
\dot \psi & \le & \frac{\sqrt{\theta p'}}{\sqrt{\eta+\theta q' \|x\|^2}}(a \|x\| + b \|e\|)\\
\nonumber &&  + \frac{\sqrt{\theta p'} \|e\|}{2(\eta+\theta q' \|x\|^2)^{3/2}} ( \lambda \eta -q'  \|x\|^2 + p'  \|e\|^2+ 
2 \theta q' a \|x\|^2 +2 \theta q' b \|x\| \|e\|) \\
\nonumber & \le & a \sqrt{p'/q'} +b \psi + \frac{1}{2\theta} \psi^3 + b\sqrt{q'/p'} \psi^2  \\
\nonumber && + \frac{\sqrt{\theta p'} \|e\|}{2(\eta+\theta q' \|x\|^2)^{3/2}} ( \lambda \eta -q'  \|x\|^2 + 
2 \theta q' a \|x\|^2 ) \\
\nonumber & \le & a \sqrt{p'/q'} +b \psi + \frac{1}{2\theta} \psi^3 + b\sqrt{q'/p'} \psi^2 + \frac{\lambda}{2} \psi \\
\label{eq:psi} &&+  \frac{\theta q'\|x\|^2}{2(\eta+\theta q' \|x\|^2)} ( -\lambda  - 1/\theta   + 
2   a) \psi
\end{eqnarray}
Hence, if $\theta\le 1/(2a-\lambda)$ (second case of the proposition) then
$$
\dot \psi \le  a \frac{p}{\sigma q} + (\frac{\lambda}{2} +b) \psi  + b\frac{\sigma q}{p} \psi^2 +  \frac{1}{2\theta} \psi^3 
$$
Then, by the Comparison Lemma, it follows that the time needed by $\psi$ to go from $0$ to $1$ is at least
$$
\int_0^1 \frac{1}{\frac{p}{\sigma q} + (\frac{\lambda}{2} +b) s + b\frac{\sigma q}{p} s^2 +  \frac{1}{2\theta} s^3 } ds.
$$ 
It can be seen that this integral is increasing with respect to $\theta\in (0,1/(2a-\lambda)]$. 
For $\theta=1/(2a-\lambda)$, this integral coincides with that given in (\ref{eq:tau2}).
Then, following the discussion in Remark~\ref{remark:dyn}, we have that (\ref{eq:tau2}) provides also 
a valid lower bound for the minimum inter-execution 
for all $\theta\in [0,1/(2a-\lambda)]$. Hence, the second case of the proposition is proved.
If $\theta > 1/(2a-\lambda)$ (third case of the proposition) then it follows from (\ref{eq:psi}) that
$$
\dot \psi \le  a \frac{p}{\sigma q} + (a+b) \psi  + b\frac{\sigma q}{p} +  \frac{1}{2\theta} ( \psi^3-\psi) 
$$
Then, by the Comparison Lemma, it follows that a lower bound on the time needed by $\psi$ to go from $0$ to $1$ is given
by (\ref{eq:tau3}).
\end{proof}

\subsection{Proof of Proposition \ref{pro:param1}}

\begin{proof}
It follows from $Q\ge \kappa P$ and
(\ref{eq:diffWlin}) that for all $t\in \R_0^+$,
$$
\frac{d}{dt}W(x(t),\eta(t)) \le (\sigma-1)\kappa W(x(t),\eta(t)).
$$
Then, for all $t\in \R_0^+$,
\begin{equation}
\label{eq:W}
V(x(t))\le W(x(t),\eta(t)) \le W(x(0),\eta(0))e^{(\sigma-1)\kappa t} =   V (x(0))e^{(\sigma-1)\kappa t}
\end{equation}
\end{proof}

\subsection{Proof of Proposition \ref{pro:param2}}                                  
\begin{proof}
If $\theta\ne 0$, then it follows from (\ref{eq:dlyaplin}) and 
Remark \ref{remark:dyn} that for all $t\in \R_0^+$,
$$
\frac{d}{dt}V(x(t)) \le (\sigma-1)x(t)^\top Q x(t) + \frac{1}{\theta} \eta(t).
$$
Then, equation (\ref{eq:W}) yields
\begin{eqnarray*}
\frac{d}{dt}V(x(t)) & \le & (\sigma-1)x(t)^\top Q x(t) + \frac{1}{\theta} \left(V(x(0)) e^{(\sigma-1)\kappa t}-V(x(t)) \right)\\
& \le & -\left(\frac{1/\chi+\theta (1-\sigma)}{\theta} \right) x(t)^\top Q x(t) + \frac{1}{\theta} V(x(0)) e^{(\sigma-1)\kappa t}.
\end{eqnarray*}
Then, integrating on both sides of the inequality yields
$$
-V(x(0)) \le -\left(\frac{1/\chi+\theta (1-\sigma)}{\theta} \right) J(x(0)) + \frac{1}{\theta\kappa(1-\sigma)} V(x(0))
$$
which is equivalent to (\ref{eq:perfdynamiclin}).
If $\theta=0$, then (\ref{eq:decaylin}) gives 
$
x(t)^\top Q x(t) \le \chi V(x(t)) \le \chi V (x(0))e^{(\sigma-1)\kappa t}.
$
Then, integrating on both sides of the inequality yields
\begin{equation}
\label{eq:perfdynamiclin0}
J(x(0)) \le \frac{\chi}{\kappa (1-\sigma)} V(x(0))
\end{equation}
which coincides with  (\ref{eq:perfdynamiclin}) for $\theta=0$.
\end{proof}
                                   
\end{document}